\documentclass{llncs}

\usepackage{graphicx}
\usepackage{amsfonts,amsmath}
\usepackage{subfigure}
\usepackage{graphicx}
\usepackage{verbatim}
\usepackage{times}
\usepackage{url}
\usepackage{multirow}
\usepackage{wrapfig}
\usepackage{placeins}
\newcommand{\rar}{\rightarrow}
\newcommand{\Z}{\mathbb Z}
\renewcommand{\O}{\mathcal O}

\let\doendproof\endproof
\renewcommand\endproof{~\hfill\qed\doendproof}

\begin{document}

\title{Happy Edges: Threshold-Coloring of Regular Lattices}
\author{Md.~J.~Alam, S.~G.~Kobourov, S.~Pupyrev, and J.~Toeniskoetter}
\institute{Department of Computer Science, University of Arizona, Tucson, USA}
\maketitle

\begin{abstract}
We study a graph coloring problem motivated by a fun Sudoku-style puzzle.
Given a bipartition of the edges of a graph into {\em near} and {\em far}
sets and an integer threshold $t$, a {\em threshold-coloring} of the graph is an assignment of integers to the vertices so
that endpoints of near edges differ by $t$ or less, while endpoints of
far edges differ by more than $t$.
We study threshold-coloring of tilings of the plane by regular polygons, known as Archimedean lattices, and their duals, the Laves lattices. We prove that some are threshold-colorable
with constant number of colors  for any edge labeling, some require an
unbounded number of colors for specific labelings, and some are not threshold-colorable.
\end{abstract}


\section{Introduction}
A Sudoku-style puzzle called \emph{Happy Edges}. Similar to Sudoku, Happy Edges is a grid
(represented by vertices and edges), and the task is to fill in the vertices with numbers so as to make all the edges ``happy''
: a solid edge is happy if the corresponding numbers of its endpoints differ by at most $1$, and a dashed
edge is happy if the numbers of its endpoints differ by at least 2; see Fig.~\ref{fig:puzzle}.

\begin{figure}[b]
\vspace{-.1cm}
\centering
	\hspace{.7cm}
	\includegraphics[scale=.65, trim={5.7cm 19.5cm 12.5cm 5cm}]{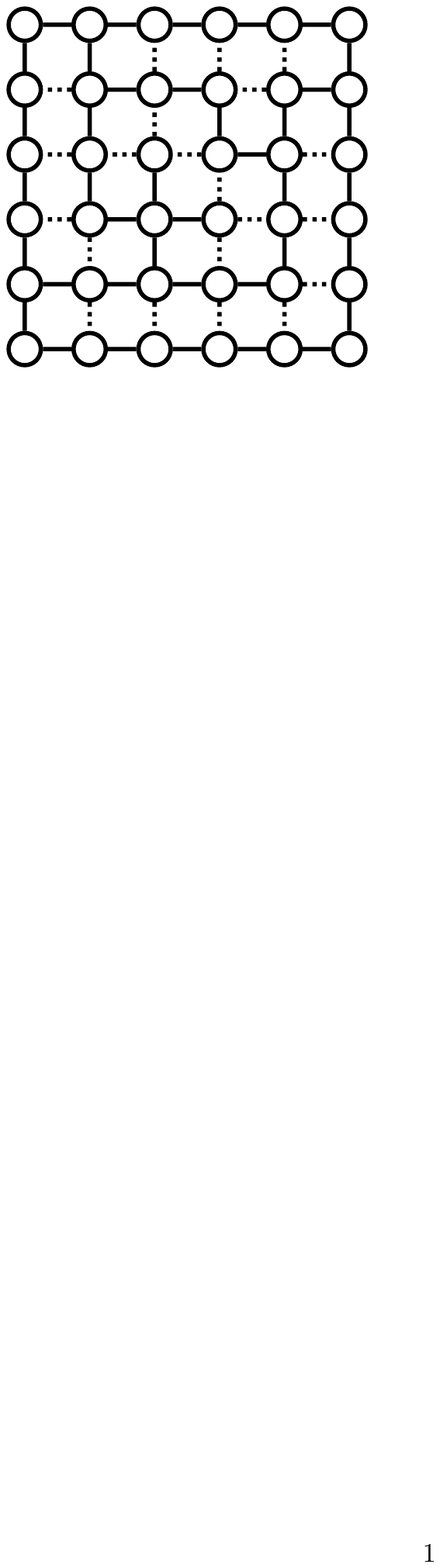}
	\label{fig:puzzle}
	
~~~~~~~~~~~~~~~~~~~~~
	\caption{An example of the \emph{Happy Edges} puzzle: fill in numbers so
          that nodes separated by a solid edge differ by at most $1$
          and nodes separated by a dashed edge differ by at least $2$.
    Fearless readers are invited to solve the puzzle before reading further!
    More puzzles are available online at \protect\url{http://happy-edges.cs.arizona.edu}.}
\end{figure}

In this paper, we study a generalization of the puzzle modeled by a graph coloring problem.
The generalization is twofold. Firstly, we consider several underlying regular grids as a base
for the puzzle, namely Archimedean and Laves lattices. Secondly, we allow for arbitrary integer
difference to distinguish between solid and dashed edges. Thus, the formal model of the puzzle
is as follows. The input is a graph with \emph{near} and \emph{far} edges. The goal is
to assign integer labels (or colors) to the vertices and compute an
integer ``threshold'' so that the
distance between the endpoints of a near edge is within the threshold, while the distance between the endpoints
of a far edge is greater than the threshold.

We consider a natural class of graphs 
called Archimedean and
Laves lattices, which yield symmetric and aesthetically appealing
game boards; see Fig.~\ref{tab:bla}. An Archimedean
lattice is a graph of an edge-to-edge tiling of the plane using regular polygons with the property
that all vertices of the polygons are identical under translation and rotation. Edge-to-edge means
that each distinct pair of edges of the tiling intersect at a single endpoint or not at all. There are
exactly $11$ Archimedean lattices and their dual graphs are the Laves
lattices (except for $3$ duals which are Archimedean). We are
interested in identifying the lattices that can be appropriately colored
for any prescribed partitioning of edges into near and far. Such lattices can be safely utilized for
the \emph{Happy Edges} puzzle, as even the simplest random strategy may serve as a puzzle generator.

\begin{figure}
\vspace{-.25cm}
\hspace{-.2cm}\includegraphics{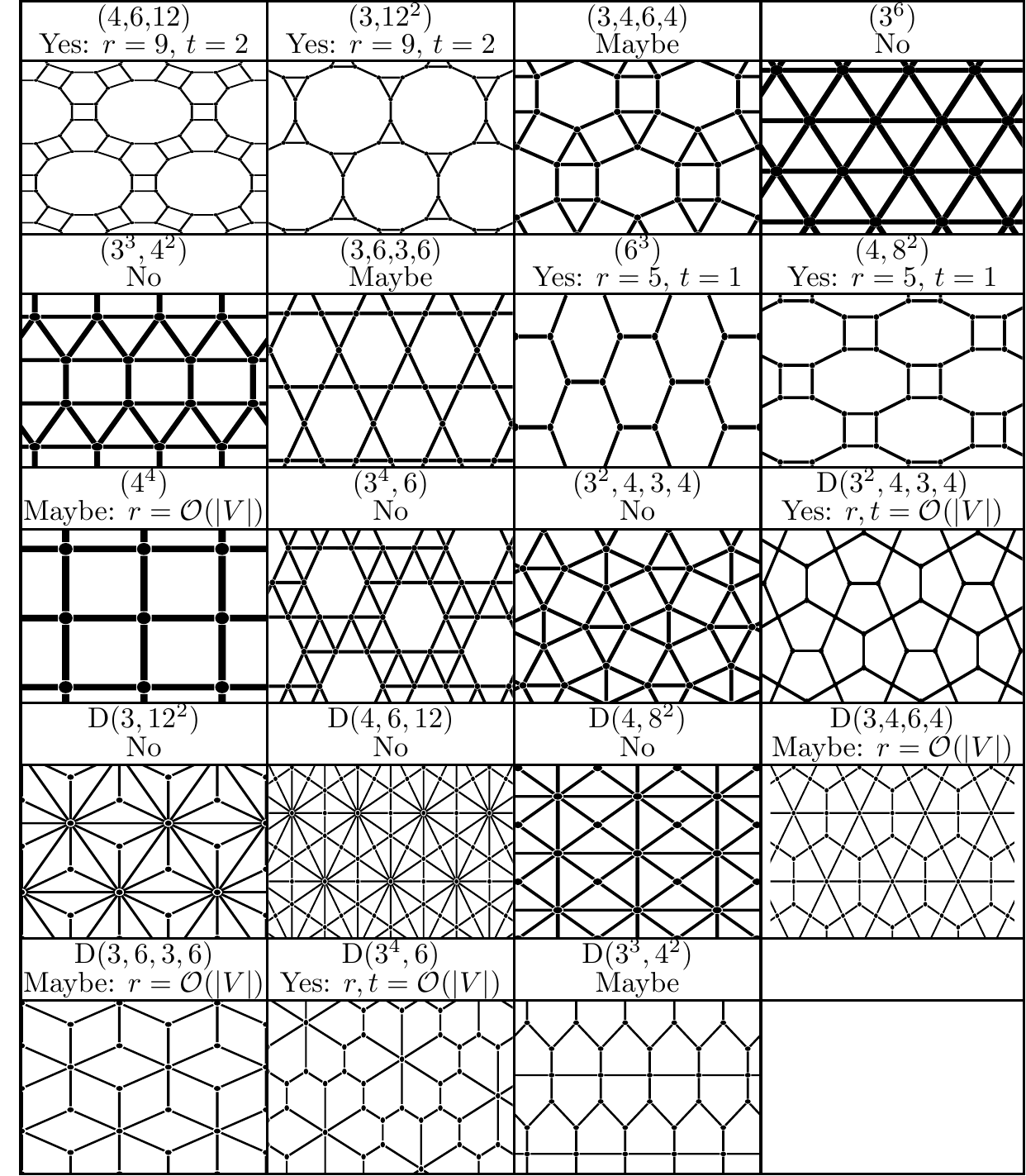}
\caption{The 11 Archimedean and 8 Laves lattices. With each lattice's name, we provide a summary of results
		concerning the threshold-coloring of the lattice. For those which are total-threshold-colorable we list
		the best known values of $r$ and $t$. For those which might be total-threshold-colorable, we list known
		constraints on $r$ and $t$.}
	\label{tab:bla}
\end{figure}

Another motivation for studying the threshold coloring problem comes from the geometric problem of
\emph{unit-cube proper contact representation} of planar graphs.
In such a representation, vertices are represented by unit-size cubes, and
edges are represented by common boundary of non-zero area between the two corresponding cubes. Finding
classes of planar graphs with unit-cube proper contact representation was posed as an
open question by Bremner~\emph{et~al.}~\cite{Bremner12}. As shown in~\cite{alam13},
threshold-coloring can be used to find such a representation of certain graphs.

\smallskip\noindent{\bf Terminology and Problem Definition:}
An \emph{edge labeling} of a graph $G=(V,E)$ is a map $l:E\rar \{N,F\}$. If $(u,v)\in E$, then
$(u,v)$ is called \emph{near} if $l(u,v)=N$ and $u$ is said to be {\em near} to $v$. Otherwise, $(u,v)$ is called \emph{far} and $u$ is {\em far} from $v$. A \emph{threshold-coloring}
of $G$ with respect to $l$ is a map $c:V\rar\Z$ such that there exists an integer $t\geq0$, called the {\em threshold},
satisfying for every edge $(u,v)\in E$, $|c(u)-c(v)|\leq t$ if and only if $l(u,v)=N$. If $m$ is the
minimum value of $c$, and $M$ the maximum, then $r> M-m$ is the \emph{range} of $c$.
The map $c$ is called a $(r,t)$\emph{-threshold-coloring} and $G$ is {\em threshold-colorable} or {\em $(r,t)$-threshold-colorable
with respect to $l$.}

If $G$ is $(r,t)$-threshold-colorable with respect to every edge labeling, then $G$ is {\em $(r,t)$-total-threshold-colorable}, or simply
{\em total-threshold-colorable.} If $G$ is not
$(r,t)$-total-threshold-colorable, then $G$ is {\em non-$(r,t)$-total-threshold-colorable}, or {\em non-total-threshold-colorable}
if $G$ is non-$(r,t)$-total-threshold-colorable for all values of $(r,t)$.

In an edge-to-edge tiling of the plane by regular polygons,
the \emph{species} of a vertex $v$ is
the sequence of degrees of polygons that $v$ belongs to, written in clockwise order. For example,
each vertex of the triangle lattice has 6 triangles, and so has species $(3,3,3,3,3,3)$. A vertex of
the square lattice has species $(4,4,4,4)$, and vertices of the
octagon-square lattice have species $(4,8,8)$. Exponents are used to
abbreviate this: $(4,8^2)=(4,8,8)$.
The {\em Archimedean tilings} are the 11 tilings by regular polygons such that each vertex has
the same species; we use this species to refer to the lattice. For example, $(6^3)$ is the hexagon lattice, and $(3,12^2)$ is
the lattice with triangles and dodecagons.
An {\em Archimedean lattice} is an infinite graph defined by the edges and vertices of an Archimedean tiling. If $A$ is an Archimedean lattice, then we refer to its dual graph as D$(A)$.
The lattice $(3^6)$ of triangles and the lattice $(6^3)$ of hexagons are dual to each other, whereas the lattice $(4^4)$ of squares is dual to itself.
The duals of the other 8 Archimedean lattices are not Archimedean, and these are referred to as {\em Laves lattices};
see Fig.~\ref{tab:bla}. By an abuse of notation, any induced subgraph of an Archimedean or Laves lattice is called an
Archimedean or Laves lattice.


\newcommand{\TTC}{\textbf{Total-Threshold-Coloring}}

\smallskip\noindent{\bf Related Work:}
Many problems in graph theory deal with coloring or labeling the vertices of a
graph~\cite{Roberts91} and many graph classes are defined based on such a coloring~\cite{Brandstadt99}.
Alam~\textit{et al.}~\cite{alam13} introduce threshold-coloring and show that deciding whether a graph
is threshold colorable with respect to an edge labeling is equivalent to the graph sandwich problem for
proper-interval-representability, which is NP-complete~\cite{Golumbic95}. They also show that graphs
with girth (that is, length of shortest cycle) at least $10$ are always total-threshold-colorable.

Total-threshold-colorable graphs are related to threshold and difference graphs.
In \emph{threshold graphs} there exists a real number $S$ and for every vertex
$v$ there is a real weight $a_v$ so that $(v,w)$ is an edge if and only if
$a_v + a_w \ge S$~\cite{Mahadev95}.
A graph is a \textit{difference graph} if there is a real number $S$ and for every vertex $v$
there is a real weight $a_v$ so that $|a_v| < S$ and $(v,w)$ is an edge if and only if
$|a_v - a_w| \ge S$~\cite{Hammer90}. Note that for
both these classes the existence of an edge is determined wholly by the threshold $S$, while
in our setting the edges defined by the threshold must also belong to the original (not necessarily
complete) graph.

Threshold-colorability is related to the \emph{integer distance graph} representation~\cite{Eggleton86,Ferrara05}.
An integer distance graph is a graph with the set of integers as vertex set and with an edge joining two vertices $u$ and $v$ if and only if $|u - v| \in D$,
where $D$ is a subset of the positive integers. Clearly, an integer distance graph is threshold-colorable if the set $D$ is a set of consecutive
integers. Also related is \emph{distance constrained graph
labeling}, denoted by $L(p_1,\dots, p_k)$-labeling, a labeling of the vertices of a graph
so that for every pair of vertices with distance at most $i\le k$ the difference of their
labels is at least $p_i$. $L(2,1)$-labelings are
well-studied~\cite{Fiala05} and minimizing the number
of labels is NP-complete, even for diameter-2 graphs~\cite{Griggs92}. It is NP-complete
to determine if a labeling exists with at most $k$ labels for every fixed integer $k\ge 4$~\cite{Fiala01}.

\smallskip\noindent{\bf Our Results:} We study the threshold-colorability of
the Archimedean and Laves lattices; see Fig.~\ref{tab:bla} for an overview
of the results. First, we prove that $6$ of them are
threshold-colorable for any edge labeling. Hence, the \emph{Happy Edges}
puzzle always have a solution on these lattices.
Then we show that $7$ of the
lattices have an edge labeling admitting no threshold-coloring. Finally, for $3$ no
constant range of colors suffices.
The puzzle motivating the problem is available at \url{http://happy-edges.cs.arizona.edu}; see
also Fig.~\ref{fig:game}.


\section{Total-Threshold-Colorable Lattices}
Given a graph $G = (V,E)$, a subset $I$ of $V$ is called \emph{2-independent} if the
shortest path between any two distinct vertices of $I$ has length at least 3. For a
subset $V'$ of $V$, we denote the subgraph of $G$ induced by $V'$ as $G[V']$. We give
an algorithm for threshold-coloring graphs whose vertex set has a partition into
a 2-independent set $I$ and a set $T$ such that $G[T]$ is a forest.
Dividing $G$ into a forest and 2-independent set has been used for other graph coloring
problems, for example in \cite{albertson04,timmons08} for the star coloring problem.

\subsection{The ($6^3$) and $(4,8^2)$ Lattices}

\begin{lemma}
	Suppose $G = (I\cup T,E)$ is a graph such that $I$ is 2-independent, $G[T]$
	is a forest, and $I$ and $T$ are disjoint. Then $G$ is
	$(5,1)$-total-threshold-colorable.
\end{lemma}
\begin{proof}
	Suppose $l:E\rar\{N,F\}$ is an edge labeling.
	For each $v\in I$, set $c(v) = 0$. Each vertex in $T$ is assigned a color from
	$\{-2,-1,1,2\}$ as follows. Choose a component $T'$ of $G[T]$, and select a root
	vertex $w$ of $T'$. If $w$ is far from a neighbour in $I$, set $c(w)=2$. Otherwise, $c(w)=1$.
	Now we conduct breadth first search on $T'$, coloring each vertex as it is traversed. When we
	traverse to a vertex $u\neq w$, it has one neighbour $x\in T'$ which has been colored, and at most one neighbour
	$v\in I$. If $v$ exists, we choose the color $c(u)=1$ if $l(u,v)=N$, and $c(u)=2$ otherwise. Then, if the edge
	$(u,x)$ is not satisfied, we multiply $c(u)$ by $-1$.
	If $v$ does not exist, we choose $c(u)= 1$ or $-1$ to satisfy the edge $(u,x)$.
    By repeating the procedure on each component of $G[T]$, we construct a $(5,1)$-threshold-coloring
    of $G$ with respect to the labeling $l$.
\end{proof}

The ($6^3$) and $(4,8^2)$ lattices have such a decomposition; see Fig.~\ref{fig:hexoct}. Hence,

\begin{theorem}
	The $(6^3)$ and $(4,8^2)$ lattices are (5,1)-total-threshold-colorable.
\end{theorem}

\begin{figure}[t]
\vspace{-1cm}	
\centering
	\subfigure[][]{
		\includegraphics[height=2.2cm, trim=0 0 0 10, scale=.55]{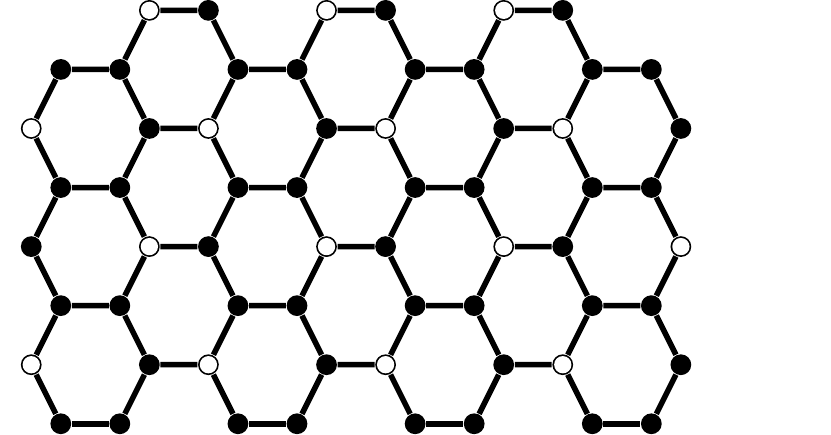}
		\label{fig:hex}
	}
~~~~~~~~~~~~~~~
	\subfigure[][]{
		\includegraphics[height=2.2cm, trim=0 0 0 10, scale=.5]{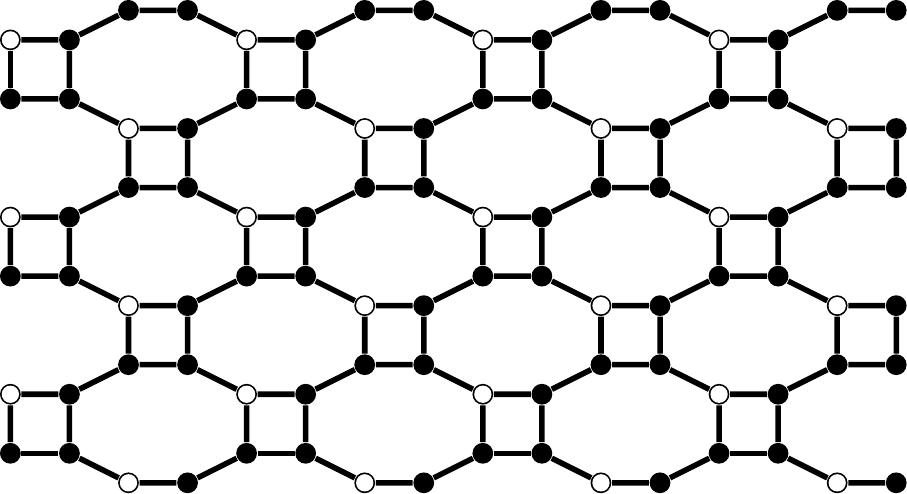}
		\label{fig:octsq}
	}
	\caption{Decomposing vertices into a 2-independent set, shown in white, and a forest. (a) The $(6^3)$ lattice.
		(b) The $(4,8^2)$ lattice.}
	\label{fig:hexoct}
\end{figure}

\subsection{The ($3,12^2)$ and $(4, 6, 12)$ Lattices}

In order to color the lattices, we use $(9, 2)$-color space, that is, threshold 2 and 9 colors, such
 as $\{0,\pm1,\pm2,\pm3,\pm4\}$. This color-space has the following properties.

\begin{lemma}\label{lemma:coloring}
  Consider a path with 3 vertices $(v_0,v_1,v_2)$, such that $v_0$,$v_2$ have
  colors $c(v_0),c(v_2)$ in $\{0,\pm1,\pm2,\pm3,\pm4\}$. For threshold 2 and any edge
  labeling,
  \begin{enumerate}
    \item[(a)] If $c(v_0) = 0,$ and $c(v_2)\in\{\pm1,\pm2,\pm3,\pm4\}$, then we can choose $c(v_1)$ in $\{\pm2,\pm3\}$.
    \item[(b)] If $c(v_0) = 0$ and $c(v_2)\in\{\pm2,\pm3,\pm4\}$, then we can choose $c(v_1)$ in $\{\pm2,\pm4\}$.
    \item[(c)] If $c(v_0) = \pm1$, and $c(v_2)\in\{\pm2,\pm3\}$, then we can choose $c(v_1)$ in $\{\pm1,\pm4\}$.
  \end{enumerate}
\end{lemma}
\begin{proof}
  \begin{enumerate}
  \item[(a)] First, we choose $c(v_1)=\pm2$ if $v_1$ is near to $v_0$, and $\pm3$
    otherwise. Then, if $v_1$ is near to $v_2$, choose the sign of $c(v_1)$ to
    agree with $c(v_2)$. Otherwise choose the sign of $c(v_1)$ to be opposite
    $c(v_2).$
  \item[(b)]Choose $c(v_1)=\pm2$ if $v_1$ is near to $v_0$, and $\pm4$ otherwise. Then,
    choose the sign of $c(v_1)$ as before.
  \item[(c)] Choose $c(v_1)=\pm1$ if $v_1$ is near to $v_0$, and $c(v_1)=\pm4$
    otherwise. Then, choose the sign of $c(v_1)$ as before.
  \end{enumerate}
\end{proof}

On a high level, our algorithms for the $(3,12^2)$ and $(4, 6, 12)$ lattices 
are very similar to each other: we identify small ``patches'', and then assemble them
into the lattice; see Figs.~\ref{fig:tiles}-\ref{fig:tridodec}.
We first show how to color a patch for $(3,12^2)$ and then for $(4,6,12)$.

\begin{lemma}
\label{lemma:tridodec}
Let $G$ be the graph shown in Fig~\ref{fig:tridodectile}.
Suppose $c(u_0) = c(u_1) = 0$ and $c(v_0) = \pm1$. Then for any edge labeling,
this coloring can be extended to a $(9, 2)$-threshold-coloring of $G$ such that $v_5$ is colored 1 or $-1$.
\end{lemma}

\begin{figure}[t]
\vspace{-1cm}	
\centering
	\subfigure[][]{
		\includegraphics{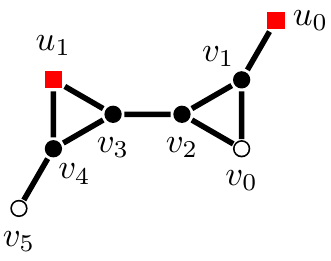}
		\label{fig:tridodectile}
	}
~~~~~~~~~~~~~~~~~~~~~~~
	\subfigure[][]{
		\includegraphics{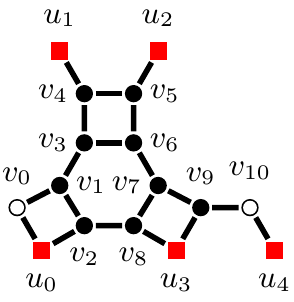}	
		\label{fig:sqhexdodectile}
	}
	\vspace{-.0cm}\caption{Illustration of Lemma~\ref{lemma:tridodec} and~\ref{lm:sqhexdodec}.
    (a)~A subgraph of the ($3,12^2$) lattice. (b)~A subgraph of the ($4,6,12$) lattice.
		Square vertices are labeled $0$.}
	\label{fig:tiles}
\end{figure}

\begin{proof}
Assume $c(v_0)=1$. We apply Lemma~\ref{lemma:coloring}(a) to the path $(u_0,v_1,v_0)$ to choose a color for $v_1$ in $\{\pm2,\pm3\}$,
then apply part (c) of the lemma to the path $(v_0,v_2,v_1)$
to choose $c(v_2)\in\{\pm1,\pm4\}$. Then $c(v_3)$ is chosen in $\{\pm2,\pm3\}$ using
part (a) of the lemma on the path $(u_1,v_3,v_2)$, and finally $c(v_4)\in\{\pm2,\pm3\}$ is chosen using part (a) on the path $(u_1,v_4,v_3)$.
Then we may choose $c(v_5)=1$ or $-1$ so that it is near or far from $c(v_4)$.
\end{proof}

A similar lemma concerns the $(4,6,12)$ lattice; see the proof in the Appendix.

\begin{lemma}
\label{lm:sqhexdodec}
	Let $G$ be the graph shown in Fig.~\ref{fig:sqhexdodectile}, and consider any edge labeling.
	Suppose that $c(u_i)=0$, for $i=0,\dots,4$, and $c(v_0)$ is a fixed color
	in $\{\pm2,\pm4\}$ that satisfies the label of $(v_0,u_0)$. Then we
	can extend this partial coloring to a coloring $c$ of all of $G$, so that $c$ is a
	(9,2)-threshold-coloring of $G$ with respect to the edge labeling, and $c(v_{10})$
	is in $\{\pm2,\pm4\}$.
\end{lemma}

\begin{theorem}
\label{thm:sqhexdodec}
The $(3,12^2)$ and $(4,6,12)$ lattices are (9,2)-total-threshold-colorable.
\end{theorem}

\begin{proof}
We prove the claim for $(3,12^2)$; see Appendix for the $(4,6,12)$  proof.

First, we join several copies of the graph $G$ in Lemma~\ref{lemma:tridodec}. Let
$G_1,\dots,G_n$ be copies of $G$. Let us call $u_{i,k}$ and $v_{j,k}$
the vertices in $G_k$, corresponding to
$u_i,v_j$ ($i=0$ or $1,0\leq j\leq 5$). For $1\leq k<n$, we set $v_{5,k}=v_{0,k+1}$. This defines a single row of the ($3,12^2$) lattice.
We can construct a $(9,2)$-threshold-coloring of this chain of $G_1,\dots,G_n$ by giving
the vertex $v_{0,1}$ the color 1 and repeatedly applying Lemma~\ref{lemma:tridodec}.

To construct the next row, we add a copy of $G$ connected to $G_i$ and $G_{i+2}$ for each odd $i$ with
$1\leq i\leq k-2$, by identifying $u_{1,i}=u_0$ and $u_{0,i+2}=u_1$. We then join the copies of $G$ added
above the first row in the same way that the copies $G_1,\dots,G_n$ were joined. By repeatedly adding new
rows, we complete the construction of the ($3,12^2$) lattice. We can threshold-color each row,
and since the rows are connected only by vertices colored 0, the
entire graph is $(9,2)$-total-threshold-colorable; see Fig.~\ref{fig:tridodec}.
\end{proof}

\begin{figure}[t]
\vspace{.2cm}
\centering
	\subfigure[][]{
		\includegraphics[height=3cm, scale=0.65]{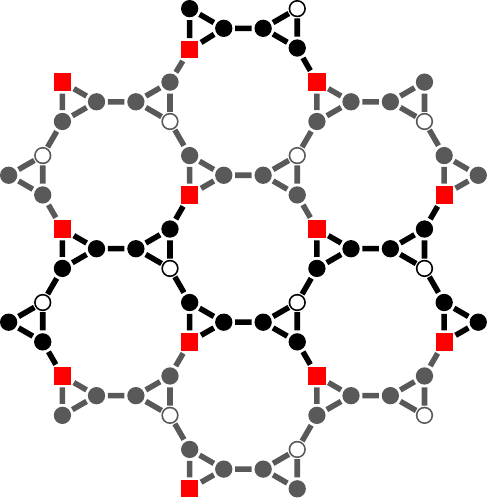}
		\label{fig:tridodec1}
	}
~~~~~~~
	\subfigure[][]{
		\includegraphics[height=3cm, scale=0.65]{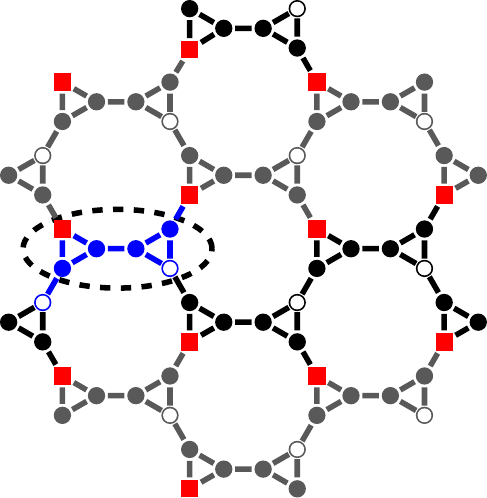}
		\label{fig:tridodec2}
	}
~~~~~~~
	\subfigure[][]{
		\includegraphics[height=3cm, scale=0.65]{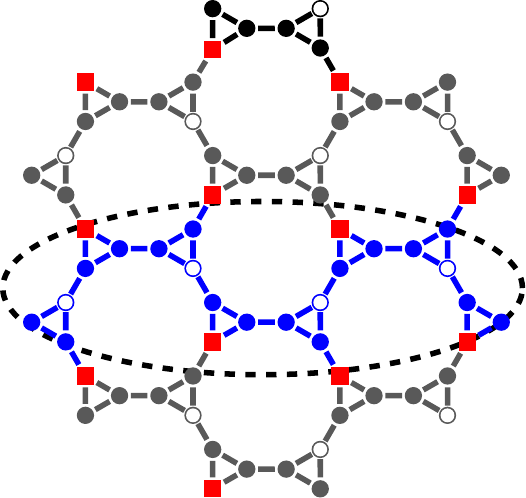}
		\label{fig:tridodec3}
	}
	\caption{Threshold-coloring the ($3,12^2$) lattice. (a) Identifying the
		rows separated by square vertices. (b) One patch has been colored, shown inside the oval. (c)
		Coloring an entire row.}
	\label{fig:tridodec}
\end{figure}

\subsection{The D($3^2,4,3,4$) and D($3^4,6$) Lattices}

Here we give an algorithm for threshold-coloring of the D($3^2,4,3,4$) and D($3^4,6$) lattices
using $\O(|V|)$ colors and $\O(|V|)$ threshold. By \emph{$k$-vertex}, we mean a vertex of degree $k$.
We use the following strategy. First, we construct an independent set $I$. For
the D($3^2,4,3,4$) lattice, $I$ consists of all the 4-vertices;
see Fig.~\ref{fig:penta1}(b).
For the D($3^4,6$) lattice, $I$ consists of all the 6-vertices and some 3-vertices;
see Fig.~\ref{fig:penta2}(a). Consider an edge labeling $l:E\rar\{N,F\}$.
We color all the vertices of $I$ using $|I|$ different colors such that each of these
vertices gets a unique color. Next we
color the remaining 3-vertices so that for each edge $e=(u,v)$ of the graph
$|c(u)-c(v)|\le|I|$ if and only if $l(e) = N$. By definition, this gives a threshold-coloring of the graph with
threshold $|I|$. Note that for both these lattices,
the 3-vertices remaining after the vertices in $I$ are removed induce a matching, that is, a set of edges with disjoint
end-vertices. We color these 3-vertices in pairs, defined by the matching.

We now describe the algorithm. Consider the graph $G_6$
with edges $e_0, \dots, e_4$ partitioned into near and far and
 coloring
 $c:\{w_1, w_2, w_3, w_4\}\rightarrow \{k+2, \ldots, 2k+1\}$ for some integer $k>0$
 such that each of the vertices gets a unique color; see Fig.~\ref{fig:penta1}(a).

 After possible renaming assume that
 if $l(e_1)\neq l(e_2)$ then $l(e_1) = N$, $l(e_2) = F$, and if $l(e_3)\neq l(e_4)$ then
 $l(e_3) = N$, $l(e_4) = F$. We say that
 $c$ is \textit{extendible} with respect to $l$ if at least one of the following conditions hold.

\begin{enumerate}
	\item $l(e_1)=l(e_2)$ or $l(e_3)=l(e_4)$; that is, at least one pair between
		$\{e_1, e_2\}$ and $\{e_3, e_4\}$ gets the same edge labeling from $l$.

	\item $l(e_0) = N$ and $c(w_1) < c(w_2)$ if and only if $c(w_3) < c(w_4)$.

	\item $l(e_0) = F$ and $c(w_1) < c(w_2)$ if and only if $c(w_3) > c(w_4)$.
\end{enumerate}

The following lemma proves that if $c$ is extendible with respect to $l$, then there is
a ($3k+1,k$)-threshold-coloring of $G_6$; see Appendix for the proof.

\begin{figure}[t]
\vspace{-.2cm}
\centering
\hspace{-.0cm}\includegraphics[height=2.8cm,width=1\textwidth]{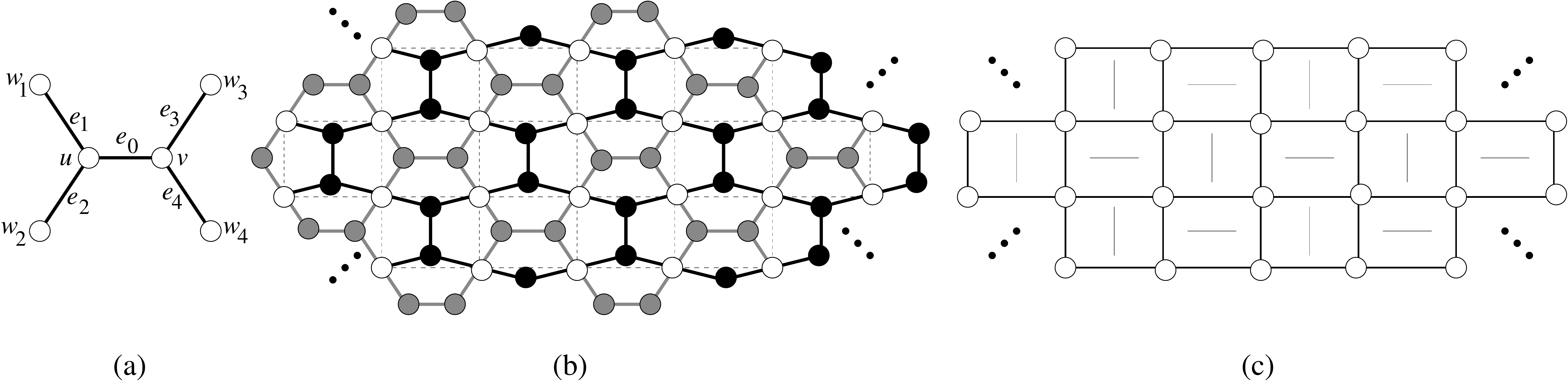}
\vspace{-.2cm}\caption{(a) The graph $G_6$, (b)--(c) Illustration for the proof of Theorem~\ref{th:penta1}.}
\label{fig:penta1}
\end{figure}

\begin{lemma}
\label{lem:linear-deg-3}
 Consider the graph $G_6$ in Fig.~\ref{fig:penta1}(a). Let $l:E\rightarrow \{N, F\}$
 be an edge labeling of $E$ and let $c:(V-\{u,v\})\rightarrow \{k+2, \ldots, 2k+1\}$
 be an extendible coloring with
 respect to $l$. Then there exist colors $c(u)$ and $c(v)$ for $u$ and $v$ from the
 set $\{1, \ldots, 3k+2\}$ such that $c$ is a threshold-coloring of $G$ for $l$
 with threshold $k$.
\end{lemma}

\begin{theorem}
\label{th:penta1}
The D($3^2,4,3,4$) lattice is (3m+2,m)-total-threshold-colorable with $m$ equal to the
number of 4-vertices in the lattice.
\end{theorem}

\begin{proof}
Let $G$ be a subgraph of D($3^2,4,3,4$) and let $l$ be an
edge labeling of $G$. Let $m$ be the
number of 4-vertices in $G$. Assign the threshold $t = m$. The remaining vertices $V_2$ of $G$ have degree 3 and they form a matching.
Each edge $(u,v)$ between these vertices is surrounded by exactly four 4-vertices,
which are the other neighbors of $u$ and $v$; see Fig.~\ref{fig:penta1}(b). Call this edge
\textit{horizontal} if it is drawn horizontally in Fig.~\ref{fig:penta1}(b); otherwise call it
\textit{vertical}. Our goal is to color the vertices of $V_1$ so that for each horizontal
and vertical edge of $G$, this coloring is extendible with respect to $l$.

Consider only the 4-vertices $V_1$ of $G$ and add an edge between two of them if
 they have a common neighbour in $G$. This gives a square grid $H$; see Fig.~\ref{fig:penta1}(c).
 Each square $S$ of $H$ is \textit{horizontal} (\textit{vertical}) if
 it is associated with a horizontal (vertical) edge in $G$.
 Let $u_1$, $u_2$, $u_3$ and $u_4$ be the left-top, right-top, left-bottom and right-bottom
 vertices of $S$ and let $c_1$, $c_2$, $c_3$ and $c_4$ be the colors assigned to them.
 Suppose $S$ is a vertical square. Then in
 order to make the coloring extendible with respect to $l$, we need that $c_1<c_2$ or
 $c_1>c_2$ implies exactly one of the two relations $c_3<c_4$ and $c_3>c_4$, depending
 on the edge-label of the associated vertical edge. Similarly if $S$ is a horizontal square then
 the relation between $c_1$ and $c_3$ implies a relation between $c_2$ and $c_4$ depending
 on the edge label of the associated horizontal edge. Consider that an edge in $H$ is directed
 from the vertex with the smaller color to the vertex with the larger color. Then for the coloring
 to be extendible to $l$, we need that for a vertical square $S$ the direction of the edge
 $(u_3,u_4)$ is the same as or opposite to that of $(u_1, u_2)$ and for a horizontal edge the
 direction of $(u_2,u_4)$ is the same as or opposite to that of $(u_1,u_3)$, depending on the
 edge-label of the associated vertical or horizontal edge. We call this a \textit{constraint}
 defined on $S$. We now show how to find an acyclic orientation of $H$ so that the constraints defined
 on the squares are satisfied.

We traverse the square grid $H$ from left-top to right-bottom. We thus assume that when
 we are traversing a particular square $S$, the orientations of its top and left edge have already
 been assigned. We now orient the bottom and right edge so that the
 constraint defined on $S$ is satisfied. We also maintain an additional invariant that
 the right-bottom vertex of each square is either a source or a sink; that is, the incident edges
 are either both outgoing or both incoming. Consider the traversal of a particular square $S$.
 If $S$ is vertical, then the direction of the bottom edge is defined by the direction of the
 top-edge and the constraint for $S$.
 We then orient the right edge so that the right-bottom vertex is either a source
 or a sink; that is, we orient the right edge upward (downward resp.) if the bottom edge is directed
 to the left (right resp.). Similarly if $S$ is horizontal, the direction of the right edge is defined
 by the constraint and we give direction to the bottom edge so that the right-bottom vertex is
 either a source or a sink. We thus have an orientation of the edges of $H$ satisfying all the
 constraints at the end of the traversal. It is easy to see that this orientation defines a directed
 acyclic graph. For a contradiction assume that there is a directed cycle $C$ in $H$. Then take
 the bottommost vertex $x$ of $C$ which is to the right of every other bottommost vertex. Then
 $x$ is either a source or a sink by our orientation and hence
 cannot be part of a directed cycle, a contradiction.

Once we have the directed acyclic orientation of $H$, we compute the coloring
 $c:V_1\rightarrow\{1,\ldots,m\}$ of the vertices $V_1$ of $G$ in a topological
 sort of this directed acyclic graph. We shift this color-space to $\{m+2, \ldots, 2m+1\}$
 by adding $m+1$ to each color. This coloring
 is extendible to the edge labeling $l$ since the orientation satisfies all the constraints.
 Thus by Lemma~\ref{lem:linear-deg-3} we can color all the 3-vertices of $G$,
 taking $k=m$. We thus have a threshold-coloring of $G$ with $3m+2$ colors and a threshold
 $m$.
\end{proof}

A similar result holds for the D($3^4,6$) lattice; see Appendix for the proof.
\begin{theorem}
\label{th:penta2}
The D$(3^4,6)$ lattice is total-threshold-colorable with $\O(|V|)$
threshold and $\O(|V|)$ colors, where $V$ is the vertex set.
\end{theorem}

\section{Non-Total-Threshold-Colorable Lattices}
In this section, we consider several lattices that cannot be threshold-colored.
 We begin with a useful lemma.

\begin{lemma}
\label{lemma:splitting}
	Consider a $K_3$ defined on $\{v_0,v_1,v_2\}$ and a 4-cycle $(u_0,u_1,u_2,u_3,u_0)$.
	Then for a given threshold $t$, a threshold-coloring $c$ and edge labeling $l$:
	\begin{enumerate}
		\item[(a)] Let $l(v_0,v_2)=F$ and $l(v_0,v_1)=l(v_1,v_2)=N$.
			If $c(v_0)<c(v_1),$ then $c(v_1)<c(v_2)$.
		\item[(b)] Let $l(v_0,v_2)=N$ and $l(v_0,v_1)=l(v_1,v_2)=F$.
			If $c(v_0)<c(v_1),$ then $c(v_2)<c(v_1)$.
		\item[(c)] Let $l(u_0,u_3)=l(u_2,u_3)=F$ and $l(u_0,u_1)=l(u_1,u_2)=N$.
			If $c(u_0)<c(u_3)$, then $c(u_1)<c(u_3)$ and $c(u_2)<c(u_3)$.
		\item[(d)] Let $l(u_0,u_1)=l(u_2,u_3)=F$ and $l(u_0,u_3)=l(u_1,u_2)=N$.
			If $c(u_0)<c(u_1)$, then $c(u_0)<c(u_2)$, $c(u_3)<c(u_1),$ and $c(u_3)<c(u_2)$.
	\end{enumerate}
	Note that we can replace $<$ with $>$ in each case.
\end{lemma}
\begin{proof}
  \begin{enumerate}
    \item[(a)] Suppose that $c(v_0)<c(v_1)$. Then $c(v_1)-t\leq c(v_0) < c(v_1)$. If $c(v_2)<c(v_1)$,
		then also $c(v_1)-t\leq c(v_2)< c(v_1)$, but then $|c(v_0)-c(v_2)|\leq t$, a contradiction.
		Thus $c(v_1)<c(v_2)$.
    \item[(b)] Suppose that $c(v_0)<c(v_1)$. If $c(v_2)>c(v_1)$, then $c(v_0)<c(v_1)<c(v_2)$ and $|c(v_0)-c(v_2)|\leq t$, so
      $|c(v_0)-c(v_1)|\leq t$, a contradiction. Hence, $c(v_2)<c(v_1)$.
    \item[(c)] Suppose that $c(u_0)<c(u_3)$. Then $c(u_0)<c(u_3)-t$ and $|c(u_0)-c(u_1)|\leq t$, so $c(u_1)<c(u_3)$,
		and therefore, $c(u_2)<c(u_3)+t$, so $c(u_2)$ must be less than $c(u_3)$ since $|c(u_2)-c(u_3)|>t$.
    \item[(d)] Suppose that $c(u_0)<c(u_1)$. Then $c(u_0)<c(u_1)-t$, $c(u_2)\geq c(u_1)-t$, and so $c(u_1)<c(u_2)$.
		$c(u_3)<c(u_1)$ since $|c(u_0)-c(u_3)|\leq t$.
		If $c(u_3)>c(u_2),$ then $c(u_1)-t\leq c(u_2) < c(u_3) < c(u_1)$, so $|c(u_2)-c(u_3)|\leq t$, a contradiction.
  \end{enumerate}
\end{proof}

\begin{figure}[t]
\vspace{-.2cm}	
\centering
	\subfigure[][]{
		\includegraphics[scale=1]{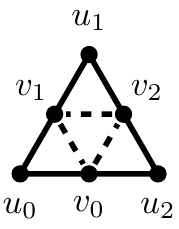}
		\label{fig:33336cex}
	}
    \hfill
	\subfigure[][]{
		\includegraphics[]{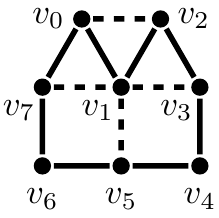}
		\label{fig:trisq}
	}
    \hfill
	\subfigure[][]{
		\includegraphics[scale=0.8]{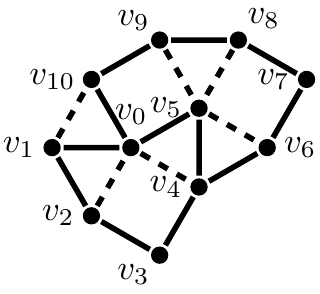}
		\label{fig:33434cex}
	}
	\hfill
	\subfigure[][]{
		\includegraphics[scale=1]{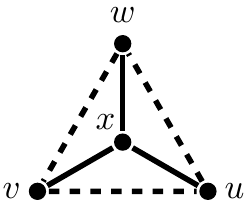}
		\label{fig:k4}
	}
	\hfill
	\subfigure[][]{
		\includegraphics[scale=0.7]{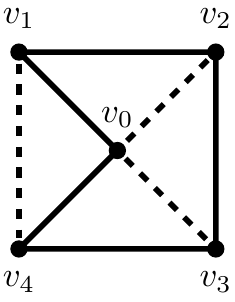}
		\label{fig:d488cex}
	}
	\caption{Non-total-threshold-colorable graphs with dashed edges labeled $F$
    and solid ones labeled $N$.
        (a)~A subgraph of ($3^6$) and ($3^4,6$).
		(b)~A subgraph of ($3^3,4^2$).
		(c)~A subgraph of ($3^2,4,3,4$).
		(d)~A subgraph of D($3,12^2$).
		(e)~A subgraph of D($4,6,12$) and D($4,8^2$).
	\label{fig:counterexamples}}
\end{figure}

\begin{theorem}
\label{thm:noncolorable}
The ($3^6$), ($3^4,6$), ($3^3,4^2$), ($3^2,4,3,4$), D($3,12^2$), D($4,6,12$), and D($4,8^2$)
lattices are non-total-threshold-colorable.
\end{theorem}
\begin{proof}
It is easy to see that a cycle with exactly 1 far edge is not $(r,0)$-threshold-colorable, so we need only prove the lattices are
not $(r,t)$-total-threshold-colorable for $t>0$. In this proof we assume that $r$ is an arbitrary integer and $t>0$.

The ($3^6$) and ($3^4,6$) lattices contain the subgraph $G$ in Fig.~\ref{fig:33336cex}.
Suppose there exists an $(r,t)$-threshold-coloring
$c$. Without loss of generality we may assume
that $c(v_0)<c(v_1)<c(v_2)$. Then $c(v_0)+t<c(v_1)$ and $c(v_1)+t<c(v_2)$,
so $c(v_0)+2t<c(v_2)$. Since the edges $(v_0,u_2)$ and $(v_2,u_2)$ are labeled $N$, we have
$|c(v_2)-c(v_0)|<|c(v_2)-c(u_2)|+|c(v_0)-c(u_2)|\leq2t$, which is a contradiction.

A subgraph of ($3^3,4^2$) is shown in Fig.~\ref{fig:trisq}. If $c$ is an $(r,t)$-threshold-coloring and w.l.o.g.
$c(v_0)<c(v_1)<c(v_2)$, then we repeatedly apply Lemma~\ref{lemma:splitting} to the vertices around the boundary. First we obtain $c(v_2)<c(v_3)$, and since $c(v_1)<c(v_3)$ we get $c(v_4)$ and $c(v_5)$ larger
than $c(v_1)$, which leads to $c(v_6)$ and $c(v_7)$ greater than $c(v_1)$. Then we must have $c(v_1)<c(v_0)<c(v_7)$, which means both $c(v_0)$ and $c(v_2)$ are in the set $\{c(v_1), c(v_1)-1,\dots, c(v_1)-t\}$,
contradicting the fact that the edge $(v_0,v_2)$ is labeled far.

For the ($3^2,3,4,3$) lattice, consider the graph in Fig.~\ref{fig:33434cex}.
Suppose there exists an $(r,t)$-threshold-coloring $c$.
Assume w.l.o.g. that $c(v_0) = 0 < c(v_1).$ By Lemma~\ref{lemma:splitting},
$c(v_2)$, $c(v_3)$, and $c(v_4)$ are positive.
Additionally, $c(v_0)<c(v_5)<c(v_4)<c(v_6)$,
and $c(v_7),c(v_8),c(v_9)$ must all be greater than $c(v_5)$. Since $c(v_5)>0$,
we have $c(v_9) \ge t+1$, and since the edge $(v_9,v_{10})$ is
labeled $N$ it must be that $c(v_{10}) > 0$. By Lemma~\ref{lemma:splitting}(a),
we have $c(v_{10})<c(v_0)<c(v_1)$, a contradiction.

D($3,12^2$) contains $K_4$ as a subgraph. Label the edges of $K_4$ so that
each edge on the outer face is far, and the other edges are near as in Fig.~\ref{fig:k4}.
Let $u,v,w$ be the vertices of the outerface, $x$ be the interior vertex, and assume an $(r,t)$-threshold-coloring $c$ exists.
Assume that $c(u)<c(x)$. From Lemma~\ref{lemma:splitting}(a), we then get
that $c(x)<c(v)$, which implies by the same lemma that $c(w)<c(x)$, and thus $c(x)<c(u)$, a contradiction.

 D($4,6,12$), and D($4,8^2$) contains the subgraph in Fig.~\ref{fig:d488cex}. Assume an $(r,t)$-threshold-coloring
 $c$ exists. Then without loss of generality say $c(v_4)<c(v_0)<c(v_1)$. By Lemma~\ref{lemma:splitting}(a)
 it follows that $c(v_1)<c(v_2)$ so $c(v_2)>c(v_0)$. By Lemma~\ref{lemma:splitting}(b) we have $c(v_3)>c(v_0)$ and
 thus $c(v_4)>c(v_0)$, a contradiction.
\end{proof}

\begin{wrapfigure}{r}{0.35\textwidth}
\vspace{-.5cm}
\centering
	\includegraphics[scale=.6]{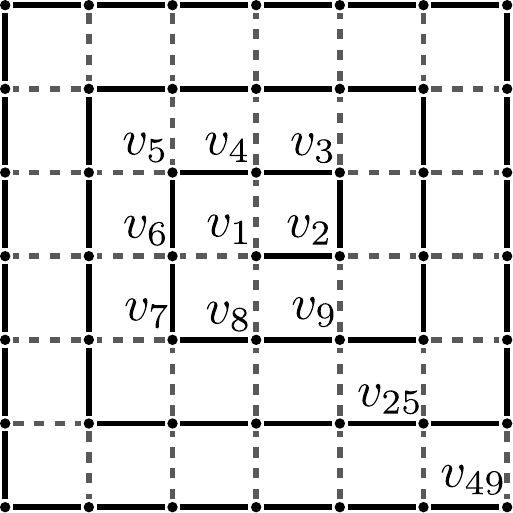}
	\caption{An example of a square lattice requiring an arbitrary number of colors. Dashed edges are far.
	\label{fig:sqspiral}}
\vspace{-1.4cm}
\end{wrapfigure}


\section{Graphs With Unbounded Colors}
We consider lattices, which are not $(r,t)$-total-threshold-colorable for any fixed
$r > 0$.

\begin{theorem}
\label{thm:346}
For every $r > 0$, there exists finite subgraphs of $(4^4)$, D$(3,4,6,4)$, and D$(3,6,3,6)$,
which are not $(r,t)$-total-threshold-colorable for any $t\geq0$.
\end{theorem}

\begin{proof}
We prove the claim for the $(4^4)$ lattice (square grid); see Appendix for the rest.

By the comment in the proof of Theorem~\ref{thm:noncolorable}, we know that the $(4^4)$ lattice is not $(r,0)$-total-threshold-colorable
for any $r$.
Let $S$ be the infinite square grid, drawn as in Fig.~\ref{fig:sqspiral}.
A vertex $v$ in $S$ has north, east, south,
and west neighbors. If $P = (v_1,\dots,v_j)$ is a path in $S$, we call $P$ a \emph{north path}
if $v_{i+1}$ is the north neighbour of $v_i$ for all $1\leq i<j$. \emph{East, south}, and \emph{west paths} are
defined similarly and these paths are uniquely defined for a given start $v_i$ and number of vertices $j$.

For each odd $n>0$, we define a path $S_n=(v_1,\dots,v_{n^2})$ in $S$. Let
$S_1$ be the path consisting of a single chosen vertex $v_1$ of $S$. Let $k = n+2$, and recursively
construct $S_k$ from $S_n$ by first adding the east neighbour $v_{n^2+1}$ of $v_{n^2}$ to $S_n$.
Then, we add the north path $(v_{n^2+1},\dots,v_{n^2+k})$, the west path $(v_{n^2+k},\dots,v_{n^2+2k})$,
the south path $(v_{n^2+2k},\dots,v_{n^2+3k})$, and the east path $(v_{n^2+3k},\dots,v_{n^2+4k})$; see Fig.~\ref{fig:sqspiral}.

With $S_n$ defined for every odd $n$, let $G_n=(V_n,E_n)$ be the subgraph of $S$ induced by the vertices
of $S_n$, and let $l_n:E_n\rar\{N,F\}$ be an edge labeling such that $l_n(e) = N$ if and only if $e$ is
in $S_n$. The graph $G_7$ is shown in Fig.~\ref{fig:sqspiral}.
We now prove that $G_n$ requires at least $n$ colors to threshold-color, for any
threshold $t>0$. W.l.o.g. suppose that $c$ is a threshold coloring such that $c(v_4)>c(v_1)$. Note that the cycles
$(v_4,v_5,v_6,v_1)$, $(v_6,v_7,v_8,v_1)$, and $(v_8,v_9,v_2,v_1)$ match the cycles in Lemma~\ref{lemma:splitting},
implying that $c(v_6),c(v_8)$ and $c(v_9)$ are greater than $c(v_1)$. This serves as the basis
for induction. Suppose that for some odd $k>1$, the vertex $c(v_{k^2}) > c(v_{(k-2)^2})$ for any
assignment $c$ of colors to the vertices of $G_n$, so long as $c(v_4) > c(v_1)$ and $c$ is an $(r,t)$-threshold-coloring for some $r>0$.
Then we consider the color $c(v_i)$, for $k^2<i\leq(k+2)^2$. There are three cases. In the first,
$v_i$ is the interior vertex of a north, east, west, or south path in $S_{k+2}$. Then $v_i$ is on a cycle
$(v_{i-1},v_i,v_j,v_{j-1})$, $j\leq k^2$, with $l(v_i,v_{i-1}) = l(v_j,v_{j-1}) = N$ and $l(v_i,v_j) =
l(v_{i-1},v_{j-1}) = F$. By Lemma~\ref{lemma:splitting}, we have $c(v_i)>c(v_j)$ and $c(v_i)>c(v_{j-1})$ so long as $c(v_{i-1})
> c(v_{j-1})$. In the second case, $v_i$ is part of a 4 cycle $(v_{i-1},v_{i},v_{i+1},v_j)$, $j\leq k^2$,
with $l(v_{i-1},v_i) = l(v_i,v_{i+1}) = N$, and the other edges labeled $F$. Again by Lemma~\ref{lemma:splitting},
we have $c(v_i)>c(v_j)$ and $c(v_{i+1})>c(v_j)$ so long as $c(v_{i-1})>c(v_j)$. The third case is the same, except $v_i$ is
in the place of $v_{i+1}$.

Given these three cases and the assumption that $c(v_{k^2})>c(v_{(k-2)^2})$, we conclude that $c(v_{(k+2)^2})
>c(v_{k^2})$ for each odd $k > 1$. Therefore, the graph $G_n$, with edge labeling $l_n$,
requires a distinct color for each of $c(v_1),c(v_{3^2}),\dots,c(v_{n^2})$.
\end{proof}

\section{Conclusion and Open Questions}

Motivated by a fun Sudoku-style puzzle, we considered the
 threshold-coloring problem for Archimedean and Laves lattices. For
 some of these lattices, we presented new coloring algorithms, while for others we found subgraphs that cannot
 be threshold-colored. Several challenging open questions remain.
 While we showed that subgraphs of the square lattice and two others require unbounded
 number of colors, we do not know whether finite subgraphs thereof are
 threshold-colorable. In the context of the puzzle, it would be useful
 to find algorithms for checking threshold-colorability for a particular subgraph of a lattice,
 rather than checking all subgraphs, as required in total-threshold-colorability.
 There are other interesting variants of the problem pertinent to the puzzle. One restricts the
 problem by allowing only a fixed number of colors to assign to the vertices. Another fixes
 the colors of certain vertices, similar to fixing boxes in Sudoku.

\bibliographystyle{abbrv}
\bibliography{thresholdFun}

\newpage
\noindent{\bf Appendix}
\begin{figure}[h]
	\centering
	\includegraphics[scale=.65, trim={4.7cm 18.5cm 11.5cm 4cm}]{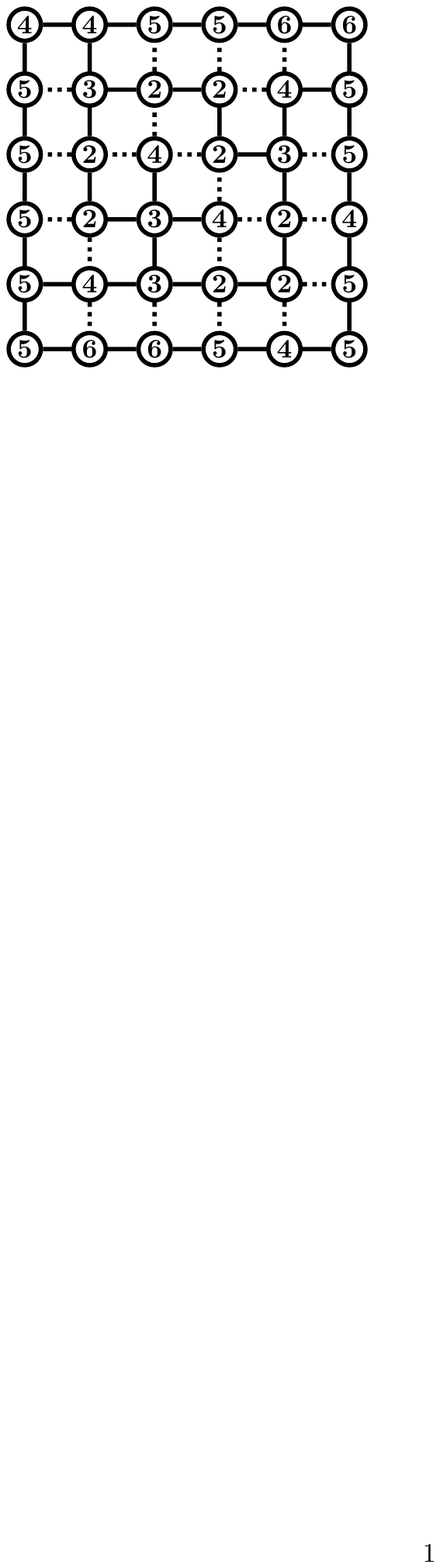}
	\label{fig:puzzle2}
	\caption{A solution to the puzzle given in Fig.~\ref{fig:puzzle}.}
\end{figure}

\smallskip\noindent Here we provide detailed proofs of theorems omitted from the body of
the paper.

\begin{figure}[h]
\vspace{-.0cm}\hspace{-.3cm}
	\subfigure[][]{
		\includegraphics[scale=0.67]{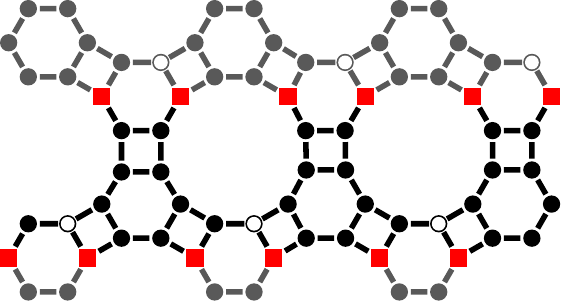}
		\label{fig:sqhexdodec1}
	}
	\subfigure[][]{
		\includegraphics[scale=0.67]{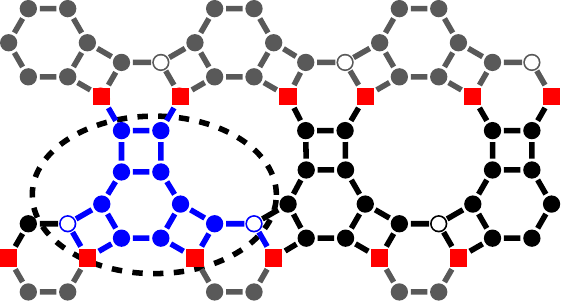}
		\label{fig:sqhexdodec2}
	}
	\subfigure[][]{
		\includegraphics[scale=0.68]{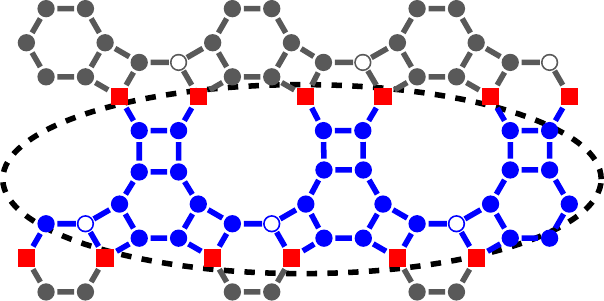}
		\label{fig:sqhexdodec3}
	}
	\caption{Threshold-coloring the ($4,6,12$) lattice. (a) Identifying the
		patches from lemma \ref{lm:sqhexdodec}. Observe that there are alternating ``rows''
		separated by square vertices. (b) One patch has been colored, shown inside the oval. (c)
		Extending the coloring to an entire row.}
	\label{fig:sqhexdodec}
\end{figure}

\begin{proof}[\textbf{Lemma~\ref{lm:sqhexdodec}}]
	Let $l$ be an edge labeling of $G$. We consider only
	the case where $c(v_0)\in \{2,4\}$ as the other case is symmetric. First, let $c(v_6)=1$.
	Using Lemma~\ref{lemma:coloring}, we color $v_5$, $v_4$, and $v_3$ so that $c(v_3)$ is in
	$\{\pm1,\pm4\}$. Consider Table~\ref{tab:color}, where we list valid colors of $v_1$ in $\{\pm2,\pm4\}$
	according to edge labeling and $c(v_3)$. An ``x'' indicates no
        color can be chosen, but in these cells we multiply $c(v_3)$ by $-1$ to obtain a color for $v_1$, and we multiply $c(v_4),c(v_5),$ and $c(v_6)$
	by $-1$ so that this is consistent.
  Use Lemma~\ref{lemma:coloring} to choose colors for $v_2$, $v_8$,
  $v_7$, $v_9$, and $v_{10}$ 
so that $v_{10}\in\{\pm2,\pm4\}$.
\end{proof}

\begin{table}
  \centering
  \begin{tabular}{|l|c|c|c|c|}\hline
    & $l(v_0,v_1)=N$ & $l(v_0,v_1)=N$ & $l(v_0,v_1)=F$ & $l(v_0,v_1)=F$ \\
	& $l(v_1,v_3)=F$ & $l(v_1,v_3)=N$ & $l(v_1,v_3)=N$ & $l(v_1,v_3)=F$ \\\hline
	
    $c(v_3)=1$ & 4   & 2   & x     & -4,-2\\
    $c(v_3)=-1$& 2,4 & x   & -2    & -4\\
    $c(v_3)=4$ & x   & 2,4 & x     & -4,-2\\
    $c(v_3)=-4$& 2,4 & x   & -4,-2 & x\\\hline
  \end{tabular}
\vspace{0.25cm}
  \caption[]{}
  \label{tab:color}
\end{table}

\begin{proof}[\textbf{Theorem~\ref{thm:sqhexdodec}}]
Let us prove the claim for the (4,6,12) lattice.

	Consider the graph $G$ from Lemma~\ref{lm:sqhexdodec}. We can construct the (4,6,12) lattice by joining together several copies of $G$. Suppose we have $k$
	copies of $G$, denoted $G_1,G_2,\dots,G_k$. If $v$ is a vertex corresponding to $v_i$ or $u_j$
	in $G_l$, then we denote $v$ by $v_{i,l}$ or $u_{j,l}$, $1\leq l\leq k$. Now, construct a row
	of the lattice from the copies of $G$ by setting $v_{10,l} = v_{0,l+1}$ and $u_{4,l}=u_{0,l+1}$.
	Lemma~\ref{lm:sqhexdodec} allows us to $(9,2)$-total-threshold-color any such chain,
	by fixing the color $v_{0,1}$ to be 2 or 4 (depending on the edge $(v_{0,1},u_{0,1})$, and then
	applying the lemma in sequence to $G_1,\dots,G_k$.
	
	Given two such chains $G_1,\dots,G_k$ and $G_1',\dots,G_{k-1}'$, we can stack the second on top
	of the first. First, if we have a vertex $v$ corresponding to $v_i$ or $u_j$ in the graph
	$G_l'$, then we denote $v$ by $v'_{i,l}$. Now, we join the two chains by
	identifying the vertex $u_{2,l}$ with $u'_{0,l}$ and the vertices $u_{1,l+1},u_{2,l+1}$ with
	$u'_{3,l}$ and $u'_{4,l}$, respectively. By repeatedly adding
        new rows, we can complete the construction of the (4,6,12)
        lattice; see Fig.~\ref{fig:sqhexdodec1}. The rows are connected only by vertices colored 0, so we
		can color each row to threshold color the entire lattice.
\end{proof}

\begin{proof}[\textbf{Lemma~\ref{lem:linear-deg-3}}]
Let $c(w_i)=c_i$ for $i\in\{1,2,3,4\}$. After possible renaming assume that
 if $l(e_1)\neq l(e_2)$ then $l(e_1) = N$, $l(e_2) = F$ and if $l(e_3)\neq l(e_4)$ then
 $l(e_3) = N$, $l(e_4) = F$. Define two integers $\lambda_{1,2}$ and
 $\lambda_{3,4}$ as follows. If $c_1<c_2$ then $\lambda_{1,2}=c_2-k-1$;
 otherwise $\lambda_{1,2}=c_2+k+1$.
 Similarly, if $c_3<c_4$ then $\lambda_{3,4}=c_4-k-1$;
 otherwise $\lambda_{3,4}=c_4+k+1$. For both
 values of $\lambda_{i,j}$ we have $|c_i-\lambda_{i,j}|\le k$;
 but $|c_j-\lambda_{i,j}|=k+1$. Also since $c_i\in\{k+2, \ldots,
 2k+1\}$, the value
 for $\lambda_{i,j}$ is in the set $\{1, \ldots, 3k+2\}$.
 Table~\ref{tab:assign} lists the colors $c(u)$ and $c(v)$ assigned to $u$ and $v$
 for all possible edge labeling $l$ assuming that the coloring for $V-\{u,v\}$ is extendible with
 respect to $l$.

\begin{table}[t]
\vspace{-.3cm}\hspace{-.2cm}
\begin{tabular}{|c|c||c|c|c|c|c|c|c|c|c|c|c|c|c|c|c|c|}

	\hline
	\multicolumn{2}{|c||}{} & 1 & 2 & 3 & 4 & 5 & 6 & 7 & 8 &
		9 & 10 & 11 & 12 & 13 & 14 & 15 & 16 \\
	\hline

	\multicolumn{2}{|c||}{$l(e_1)=$} & $N$ & $N$ & $N$ & $N$ & $N$ & $N$ & $N$ & $N$ &
		$F$ & $F$ & $F$ & $F$ & $F$ & $F$ & $F$ & $F$ \\
	\multicolumn{2}{|c||}{$l(e_2)=$} & $N$ & $N$ & $N$ & $N$ & $F$ & $F$ & $F$ & $F$ &
		$N$ & $N$ & $N$ & $N$ & $F$ & $F$ & $F$ & $F$ \\
	\multicolumn{2}{|c||}{$l(e_3)=$} & $N$ & $N$ & $F$ & $F$ & $N$ & $N$ & $F$ & $F$ &
		$N$ & $N$ & $F$ & $F$ & $N$ & $N$ & $F$ & $F$ \\
	\multicolumn{2}{|c||}{$l(e_4)=$} & $N$ & $F$ & $N$ & $F$ & $N$ & $F$ & $N$ & $F$ &
		$N$ & $F$ & $N$ & $F$ & $N$ & $F$ & $N$ & $F$ \\

	\hline\hline

	\multirow{2}{*}{\parbox{0.7cm}{$l(e_0)$ $=N$}} & $c(u)$ &
		$k+1$ & \parbox{0.8cm}{$k+1$ or $2k+2$} & \multirow{4}{*}{--} & $k+1$ &
		$\lambda_{1,2}$ & $\lambda_{1,2}$ & \multirow{4}{*}{--} & $\lambda_{1,2}$ &
		\multicolumn{4}{|c|}{\multirow{4}{*}{--}} &
		$1$ & \parbox{0.8cm}{$1$ or $3k+2$} & \multirow{4}{*}{--} & 1 \\
	\cline{2-4} \cline{6-8} \cline{10-10} \cline{15-16} \cline{18-18}
	 & $c(v)$ & $k+1$ & $\lambda_{3,4}$ &  & $1$ & \parbox{0.8cm}{$k+1$ or $2k+2$} &
		$\lambda_{3,4}$ &  & \parbox{0.8cm}{$1$ or $3k+2$} &
		 \multicolumn{4}{|c|}{}  & $k+1$ & $\lambda_{3,4}$ &  & $1$ \\

	\cline{1-4} \cline{6-8} \cline{10-10} \cline{15-16} \cline{18-18}

	\multirow{2}{*}{\parbox{0.7cm}{$l(e_0)$ $=F$}} & $c(u)$ &
		$k+1$ & \parbox{0.8cm}{$k+1$ or $2k+2$} &  & $k+1$ &
		$\lambda_{1,2}$ & $\lambda_{1,2}$ &  & $\lambda_{1,2}$ &
		\multicolumn{4}{|c|}{} & $3k+2$ & \parbox{0.8cm}{$1$ or $3k+2$} &  & 1 \\
	\cline{2-4} \cline{6-8} \cline{10-10} \cline{15-16} \cline{18-18}
	 & $c(v)$ & $2k+2$ & $\lambda_{3,4}$ &  & $3k+2$ &
		\parbox{0.8cm}{$k+1$ or $2k+2$} & $\lambda_{3,4}$ &  &
		\parbox{0.8cm}{$1$ or $3k+2$} &
		\multicolumn{4}{|c|}{} & $k+1$ & $\lambda_{3,4}$ &  & $3k+2$ \\
	\hline

\end{tabular}
\vspace{0.25cm}
\caption{Proof of Lemma~\ref{lem:linear-deg-3}: Assignment of $c(u)$ and $c(v)$ for all possible edge labeling $l$ of $G$ and coloring $c$ of $V-\{u,v\}$}
\label{tab:assign}
\end{table}

Here by our definition of $e_1$, $e_2$, $e_3$ and $e_4$, it is not possible that $l(e_1)=F$
 and $l(e_2)=N$, or $l(e_3)=F$ and $l(e_4)=N$. Again for each value of $\lambda_{i,j}$;
 $i,j = 1,2$~or~$3,4$, exactly one of the two values $k+1$ and $2k+2$
 is within distance $k$ of $\lambda_{i,j}$ and the other value is more
 than distance $k$ away. The same also holds
 for the two values $1$ and $3k+2$. Thus depending on the edge-label of $e$, there is exactly
 one valid choice for $c(u)$ and $c(v)$ in columns 2, 5, 8 and 14 of Table~\ref{tab:assign}.
 Finally in column 6, we claim that $\lambda_{1,2}$ and $\lambda_{3,4}$ are within
 distance $k$ of each other if and only if $l(e_0)=N$.
 Without loss of generality assume that $c_1<c_2$. Then
 $\lambda_{1,2}=c_2-k-1\in\{1, \ldots, k\}$ since $c_2\in\{k+2, \ldots, 2k+1\}$. Now since
 $c$ is extendible with respect to $l$, $c_3<c_4$ if and only if $l(e)=N$. Thus if $l(e)=N$,
 $\lambda_{3,4}=c_4-k-1\in\{1, \ldots, k\}$  since $c_4\in\{k+2, \ldots, 2k+1\}$. Therefore,
 $|\lambda_{1,2}-\lambda_{3,4}|<k$. On the other hand if $l(e_0)=F$, then $c_3>c_4$ and
 hence $\lambda_{3,4}=c_3+k+1\in\{2k+2, \ldots, 3k+1\}$.
 Thus $|\lambda_{1,2}-\lambda_{3,4}|>k$. Thus the assignment of $c(u)$ and $c(v)$ respects
 the edge labeling in all cases.
\end{proof}

\begin{proof}[\textbf{Theorem~\ref{th:penta2}}]
We threshold-color D($3^4,6$) with the same strategy as in
 Theorem~\ref{th:penta1}. Let $G$ be a particular instance of a D($3^4,6$) graph
 and let $l$ be a particular edge labeling of $G$. We construct an independent set
 $V_1$ of $G$ with all the 6-vertices and some 3-vertices; see
 Fig.~\ref{fig:penta2}(a). The remaining 3-vertices $V_2$ induce a matching
 in $G$. Let $m$ be the number of vertices in $V_1$. We first color these
 vertices with $m$ colors so that each of them gets a unique color. We assign
 threshold $t=m$. Each edge $(u,v)$ between two vertices from $V_2$ is surrounded
 by exactly four vertices from $V_1$ and they are the other neighbors of $u$ and $v$.
 We want to color the vertices of $V_1$ so that this coloring is extendible
 to the edge labeling $l$ for all such edges $(u,v)$, $u,v\in V_2$.

\begin{figure}[t]
\vspace{-.5cm}\centering
\includegraphics[width=0.85\textwidth]{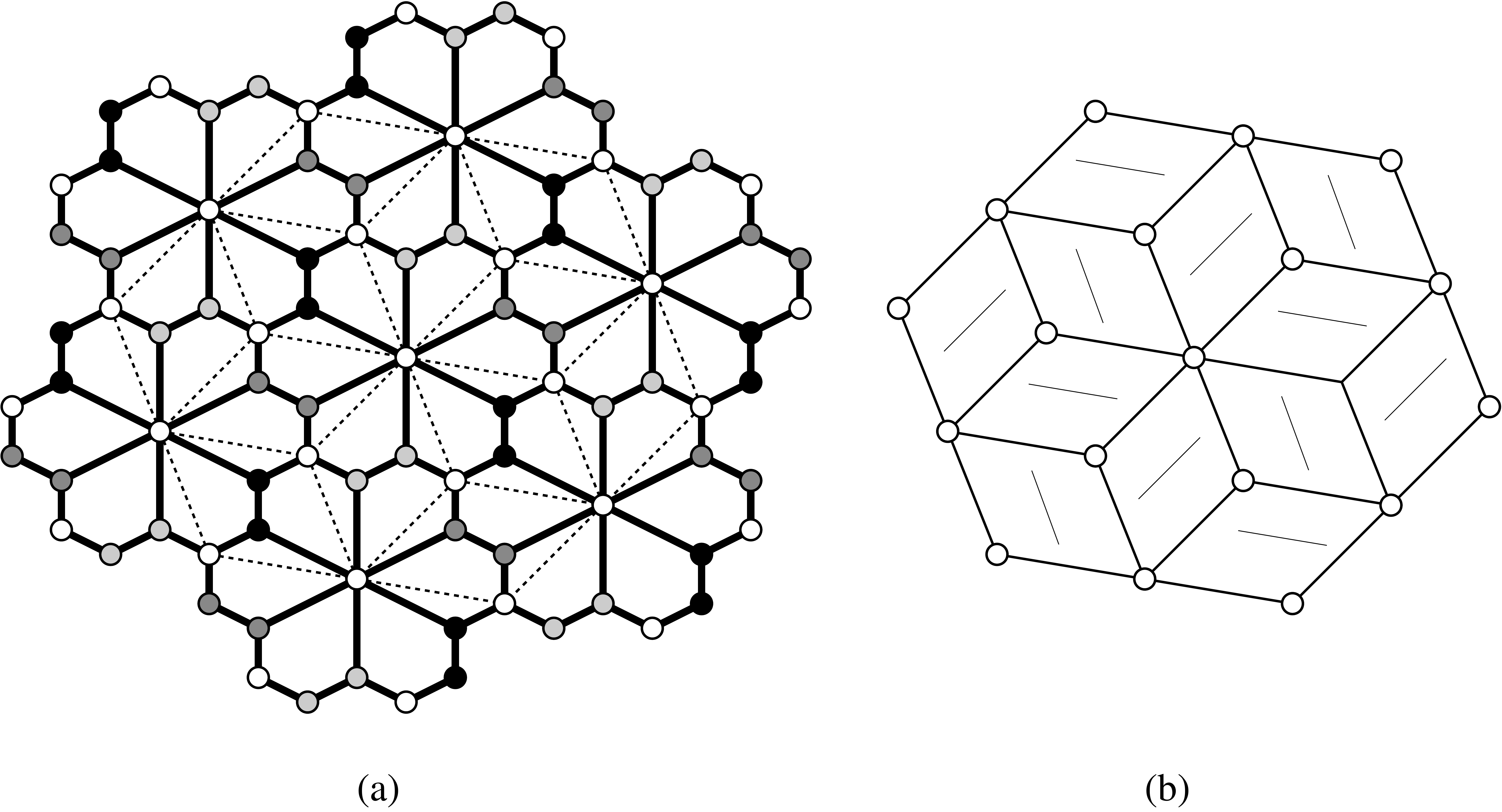}
\caption{Illustration for the proof of Theorem~\ref{th:penta2}.}
\label{fig:penta2}
\end{figure}

 We consider only the vertices of $V_1$ and add an edge between two vertices of $V_1$
 if and only if they have a common neighbour in $G$ to obtain a D($3,6,3,6$) grid $H$;
 see Fig.~\ref{fig:penta2}(b). As in Theorem~\ref{th:penta1}, if we
 direct each edge of $H$ from the vertex with smaller color to the vertex with larger color,
 the requirement of extendibility of coloring imposes some constraints on the edge directions
 between either the left and right edges or the top and bottom edges. We now
 orient all the edges in $H$ so that these constraints are satisfied and the
 graph $H$ becomes a directed acyclic graph. We traverse $H$ from left-top
 to right-bottom and
 orienting edges so that the constraints are satisfied and the right-bottom vertex for each rectangle
 in $H$ becomes either a source or sink. This orientation gives a desired directed acyclic graph
 from $H$ and we color the vertices of $V_1$ with a topological order of
 this DAG.
We again shift the color space from $\{1,\ldots, m\}$ to $\{m+2,\ldots,2m+1\}$
 and since this coloring is extendible, by Lemma~\ref{lem:linear-deg-3} we have a threshold
 coloring of $G$ with $3m+2$ colors and a threshold value of $m$.
\end{proof}

\begin{proof}[\textbf{Theorem~\ref{thm:346}}]
We first prove the claim for the D$(3,4,6,4)$ lattice.
Consider the edge labeling of D$(3,4,6,4)$ where the near edges
form a spiral-shaped path $P$; see Fig.~\ref{fig:d3464spiral}. Each face except the one with vertices labeled $v_1,v_2,v_3,v_4$
has exactly 2 near edges. We can add vertices from D$(3,4,6,4)$ to the end of this path, so that the near edges
are still a path and each newly added face has exactly 2 near edges. We label the vertices,
starting at $v_1$, using breadth first search, so that the successor of vertex $v_i$ is always
$v_{i+1}$. Suppose $c$ is a threshold-coloring
such that $c(v_1)<c(v_4)$. Now, consider a vertex $v_i$, $i>4$, with a neighbour $v_j$ such that $j<i$ and
$l(v_i,v_j)=F$. Then $v_i,v_j$ are on a face $(v_i,v_{i+1},v_k,v_j)$. There are 2 far edges; they are
either (i) $(v_i,v_j)$ and $(v_k,v_j)$ or (ii) $(v_i,v_j)$ and $(v_{i+1},v_k)$. In (i), if $c(v_k)>c(v_j)$,
then also $c(v_i)>c(v_j)$ by Lemma~\ref{lemma:splitting}, and in (ii), if $c(v_{i+1})>c(v_k)$ then
$c(v_i)>c(v_j)$, by the same lemma. Since $c(v_4)>c(v_1)$, we can inductively conclude that
$c(v_i)>c(v_j)$, and hence that as we add vertices to the path $P$, we require more and more colors.

We now prove the claim for the D$(3,6,3,6)$ lattice.
Consider the edge labeling of the D$(3,6,3,6)$ lattice shown in
Fig.~\ref{fig:d3636spiral}, with vertex set $V$. The set of near edges
forms a spiral-shaped caterpillar $T$. All but one interior face has
exactly 2 near edges. We extend
this caterpillar by adding vertices to the end of $T$, and labeling the edges so that every interior
face has exactly 2 near edges, except the one mentioned earlier. Let $l$ be this edge labeling. Now,
suppose $c$ is a threshold coloring of the caterpillar, and without loss of generality choose the
colors of $v_1,v_4$ in Fig.~\ref{fig:d3636spiral} so that $c(v_1)<c(v_4)$. We can traverse the vertices of $T$ using breadth
first search. Whenever we traverse a vertex with degree 3 in $T$, we traverse its degree 1 neighbour before we traverse its higher
degree neighbors. Assume that $v$ is a vertex with a neighbour $u$ traversed prior to $v$, such that $l(u,v)=F$ and $c(v)>c(u)$.
If $v$ is not a leaf, then $v$ is on a face $(x,y,v,u)$ either with (i) $x,y$ not traversed, and $l(x,u)=F$, or
(ii) $y$ not traversed, and $l(y,x)=F$. Since $l(u,v)=F$, the other edges are near. Then by Lemma~\ref{lemma:splitting},
we have in (i) that $c(x)>c(u)$ and in (ii) that $c(y)>c(x)$ (both since $c(v)>c(u)$). Now, since $c(v_4)>c(v_1)$,
we can traverse $T$ and conclude that each vertex $y$ with a previously traversed vertex $x$ such that $l(x,y)=F$
must have $c(y)>c(x)$. Therefore, for any $r>0$, as the caterpillar $T$ grows arbitrarily big, we require more
than $r$ colors to threshold-color it.
\end{proof}

\begin{figure}
\vspace{-.5cm}	
\centering
\subfigure[][]{
		\includegraphics[scale=.7]{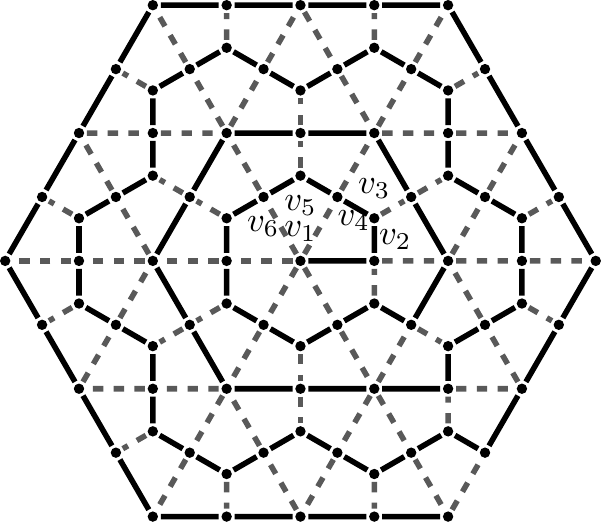}
		\label{fig:d3464spiral}
	}
~~~~~~~~~
	\subfigure[][]{
		\includegraphics[scale=.7]{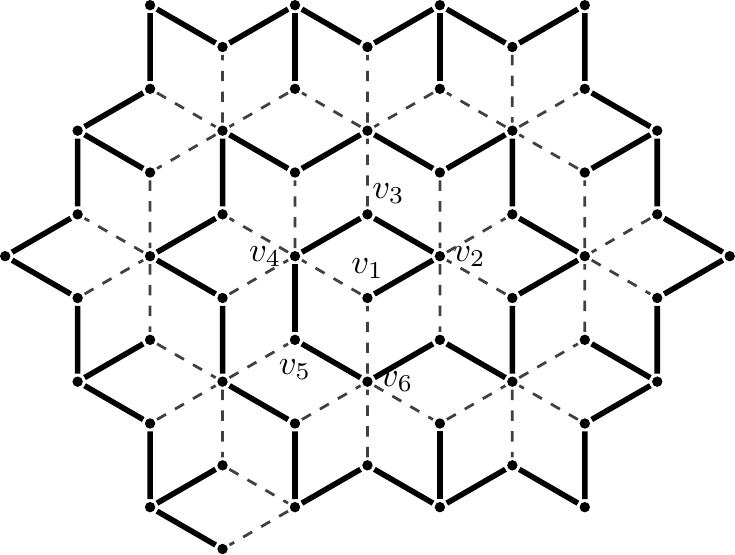}
		\label{fig:d3636spiral}
	}
	\caption[]{Subgraphs of the (a) D(3,4,6,4) lattice and (b)
		D(3,6,3,6) lattice, which require arbitrarily many colors. Dashed edges
		are labeled $F$.}
	\label{fig:spirals}
\end{figure}

\begin{figure}[t]
\centering
	\subfigure[]{
		\fbox{\includegraphics[width=.8\textwidth]{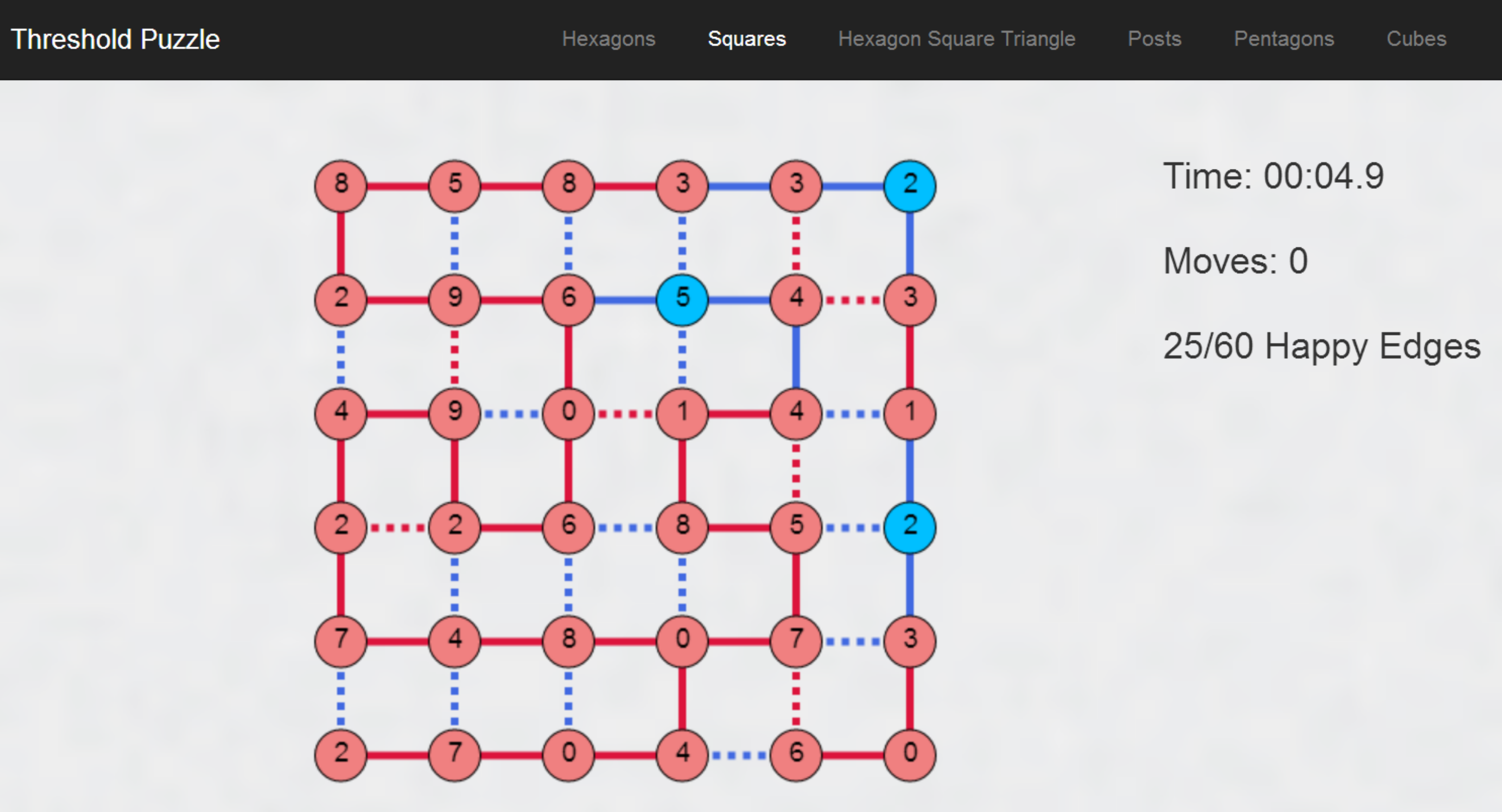}}}
	\subfigure[]{
		\fbox{\includegraphics[width=.8\textwidth]{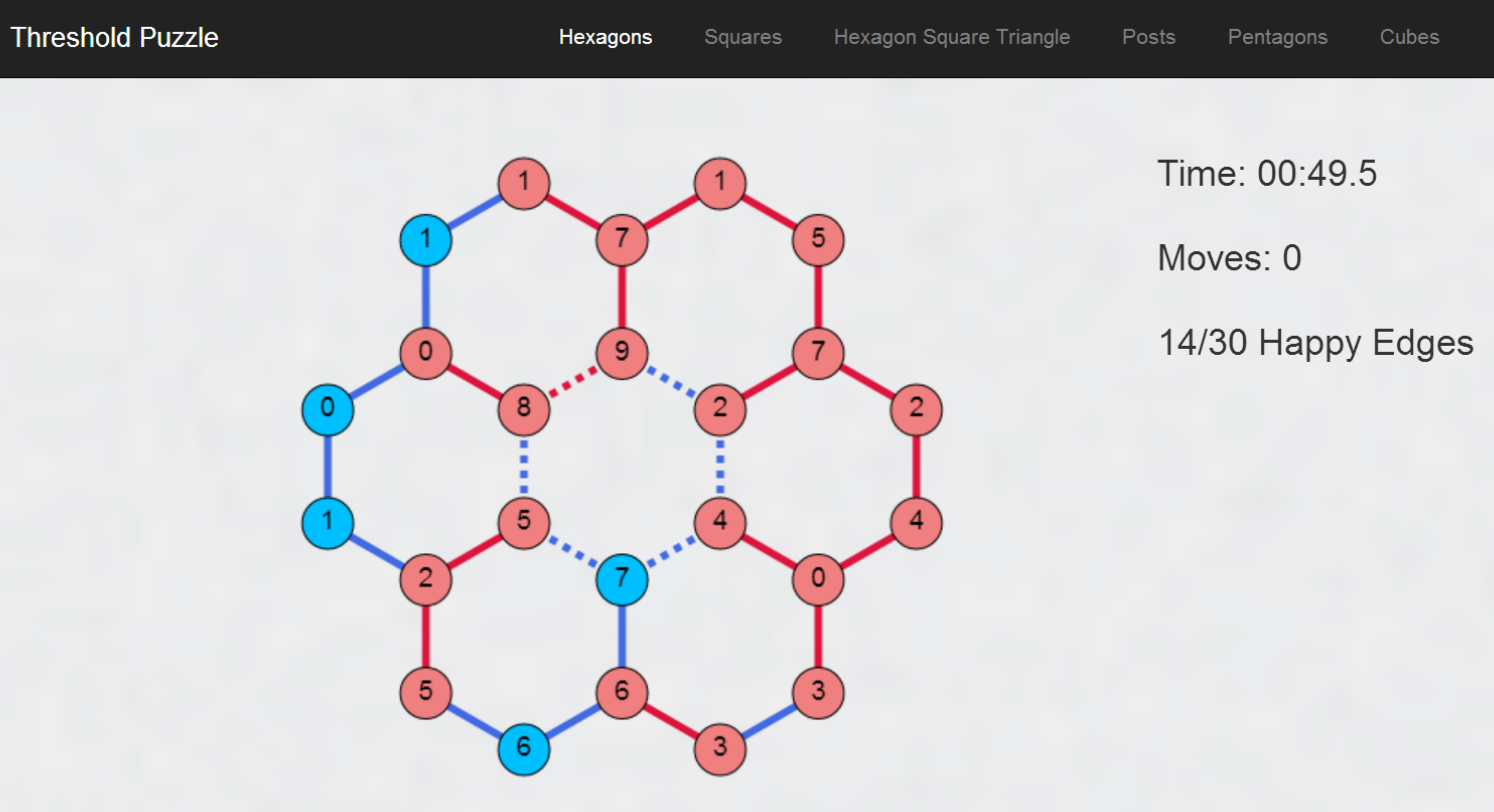}}}
	\caption{Snapshot of the \emph{Happy Edges
} puzzle environment available for
    computers and mobile devices. Dashed edges indicate F edges, and vertices/edges are blue if the
    color of their neighbors/endpoints are correct.
	(a) The square lattice. (b) The triangle square hexagon lattice.}
	\label{fig:game}
\end{figure}
\end{document}